\documentclass[a4paper,USenglish,cleveref, autoref, thm-restate]{lipics-v2021}
\pdfoutput=1 %uncomment to ensure pdflatex processing (mandatatory e.g. to submit to arXiv)
\hideLIPIcs  %uncomment to remove references to LIPIcs series (logo, DOI, ...), e.g. when preparing a pre-final version to be uploaded to arXiv or another public repository

\nolinenumbers

\usepackage{stackengine}
\usepackage{amsfonts, mathdots, amssymb, latexsym, amsbsy,bm,mathtools}         % Mathematische Fonts amslatex
\usepackage{braket}

\renewcommand\labelenumi{(\roman{enumi})}
\renewcommand\theenumi\labelenumi

\usepackage{graphicx} % Einbinden von Grafiken ( .jpg  .png  .pdf  .mps)
\usepackage[dvipsnames]{xcolor}
\definecolor{mygray}{RGB}{103, 106, 110} %188, 198, 214 %135, 139, 145
\definecolor{mycomment}{RGB}{163, 35, 80}
\definecolor{mynodecolor}{RGB}{184, 204, 237}
\definecolor{fillcolor1}{RGB}{101, 166, 143}
\definecolor{fillcolor2}{RGB}{144, 238, 144}
\definecolor{fillcolor3}{RGB}{77, 87, 129}
\definecolor{fillcolor4}{RGB}{178, 178, 198}

\usepackage{tikz}
\usetikzlibrary{positioning,arrows,shapes,trees}  
\usepackage{pgfplots}
\pgfplotsset{compat=1.8}
\usepgfplotslibrary{fillbetween}

\usepackage{float} % to use option H for the placement of figures
\usepackage{subcaption} % to have a subnumbering in the figures
\newcommand{\R}{\mathbb{R}}

\newcommand*\diff{\mathop{}\!\mathrm{d}}

\DeclarePairedDelimiterX{\scalar}[2]{\langle}{\rangle}{#1, #2}
\newcommand{\norm}[1]{\left\lVert#1\right\rVert}
\newcommand{\abs}[1]{\left\lvert #1 \right\rvert}

\newcommand{\PPAD}{\ensuremath{\mathsf{PPAD}}}

\DeclareMathOperator{\Opt}{OPT}

\renewcommand{\epsilon}{\varepsilon}

\newcommand{\source}{s}
\newcommand{\sink}{t}

\newcommand{\StratSpace}{\Sigma}

\newcommand{\horizon}{[0,H)}

\renewcommand{\paragraph}[1]{\vspace{0.5em}\noindent\textbf{#1}\;}

\title{Atomic Splittable Flow Over Time Games}
\titlerunning{Atomic Splittable Flow Over Time Games}
\author{Antonia Adamik}{Technische Universität Berlin, Germany}{antonia.mp.adamik@campus.tu-berlin.de}{}{}
\author{Leon Sering}{ETH Zurich, Switzerland}{leon@sering.eu}{https://orcid.org/0000-0003-2953-1115}{Funded by the Deutsche Forschungsgemeinschaft (DFG, German Research Foundation) under Germany's Excellence Strategy – The Berlin Mathematics Research Center MATH+ (EXC-2046/1, project ID: 390685689).}
\authorrunning{A. Adamik and L. Sering}
\Copyright{Antonia Adamik and Leon Sering}



\relatedversiondetails{Full Version}{https://arxiv.org/abs/2010.02148}

\acknowledgements{We have considered atomic splittable flow over time games in different settings and under
various assumptions in collaboration with several people. Unfortunately, most of these research directions were more
challenging than expected and not as successful as the work at hand. Nonetheless, we want to thank Laura Vargas Koch,
Veerle Timmermans, Bj\"orn Tauer, Tim Oosterwijk and Dario Frascaria for the excellent collaboration and inspiring
discussions.
}

\EventEditors{James Aspnes and Othon Michail}
\EventNoEds{2}
\EventLongTitle{1st Symposium on Algorithmic Foundations of Dynamic Networks (SAND 2022)}
\EventShortTitle{SAND 2022}
\EventAcronym{SAND}
\EventYear{2022}
\EventDate{March 28--30, 2022}
\EventLocation{Virtual Conference}
\EventLogo{}
\SeriesVolume{221}
\ArticleNo{3}
\begin{document}
\maketitle              % typeset the header of the contribution

\begin{abstract} %max 250
In an atomic splittable flow over time game, finitely many players route flow dynamically through a network, in which edges
are equipped with transit times, specifying the traversing time, and with capacities, restricting flow rates.
Infinitesimally small flow particles controlled by the same player arrive at a constant rate at the player's origin and the
player's goal is to maximize the flow volume that arrives at the player's destination within a given time horizon.
Here, the flow dynamics are described by the deterministic queuing model, i.e., flow of different players merges
perfectly, but excessive flow has to wait in a queue in front of the bottle-neck. In order to determine Nash equilibria
in such games, the main challenge is to consider suitable definitions for the players' strategies, which depend on the
level of information the players receive throughout the game. For the most restricted version, in which the players
receive no information on the network state at all, we can show that there is no Nash equilibrium in general, not even for
networks with only two edges. However, if the current edge congestions are provided over time, the players can adapt
their route choices dynamically. We show that a profile of those strategies always lead to a unique feasible flow over
time. Hence, those atomic splittable flow over time games are well-defined. For parallel-edge networks Nash equilibria
exists and the total flow arriving in time equals the value of a maximum flow over time leading to a price of anarchy
of~$1$.
\end{abstract}

%In an atomic splittable flow over time game, finitely many players route flow dynamically through a network, in which edges are equipped with transit times, specifying the traversing time, and with capacities, restricting flow rates. Infinitesimally small flow particles controlled by the same player arrive at a constant rate at the player's origin and the player's goal is to maximize the flow volume that arrives at the player's destination within a given time horizon. Here, the flow dynamics are described by the deterministic queuing model, i.e., flow of different players merges perfectly, but excessive flow has to wait in a queue in front of the bottle-neck. In order to determine Nash equilibria in such games, the main challenge is to consider suitable definitions for the players' strategies, which depend on the level of information the players receive throughout the game. For the most restricted version, in which the players receive no information on the network state at all, we can show that there is no Nash equilibrium in general, not even for networks with only two edges. However, if the current edge congestions are provided over time, the players can adapt their route choices dynamically. We show that a profile of those strategies always lead to a unique feasible flow over time. Hence, those atomic splittable flow over time games are well-defined. For parallel-edge networks Nash equilibria exists and the total flow arriving in time equals the value of a maximum flow over time leading to a price of anarchy of 1.

\section{Introduction} In \emph{static routing problems}, traffic is to be routed through a network at minimum total
cost. The cost or traveling time on each edge depends on its congestion. However, the assumption that an optimal routing
might be implemented by some superordinate authority is not realistic in many settings. More likely, each network
participant selfishly chooses a path in order to minimize their own traveling time. In general, the lack of coordination
causes a higher total traveling time. To quantify this decrease in performance, the total traveling time of a
\emph{Wardrop equilibrium} \cite{Wardrop_1952} is compared to the total traveling time of the \emph{system optimum}. The
ratio between a worst equilibrium and the system optimum is the \emph{price of anarchy} \cite{Roughgarden_Tardos_2002}.
This basic model can be extended in several ways. In this research work we want to focus on two aspects.

The first aspect is the temporal dimension. Vehicles in real traffic need time to move from the origin to the
destination and the traveling time increases with the degree of congestion, which varies over time. In other words, the
traffic flow does not traverse the network instantaneously, but progresses at a certain pace. In addition, the effects
of a routing decision in one part of the network take some time to spread across the network as a whole. In order to
mathematically model this, we add a time component, transforming static flows into \emph{flows over time}. Here, every
infinitesimally small flow particle needs time to traverse the network and the flow rates of the edges are restricted by
capacities. By assuming that each particle acts selfishly, we can consider dynamic equilibria, which are called
\emph{Nash flows over time} \cite{Koch_Skutella_2011}.

For the second aspect, note that in real-world traffic the activity of a single road user has in most cases a negligible
impact on the performance of the system as a whole. Furthermore, the assumption of independent and selfish particles is
not justifiable in all applications: Networks where participants control a flow of positive measure are not covered. For
instance, in transportation networks freight units might not act selfishly; they are controlled by freight companies
that each control a significant amount of traffic. This leads to the second aspect, cooperative behavior among groups of
network participants. To integrate this into the mathematical model, flow particles are allowed to form coalitions. In
so-called \emph{atomic splittable routing games} we consider a finite number of \emph{atomic} players (the coalitions),
each controlling a positive amount of flow volume that has to be \emph{routed} through the network but can be
\emph{split up} and divided over different routes.

In this paper we want to combine both aspects, as depicted in \Cref{fig:relationship_between_models}. That means, in
contrast to Nash flows over time, sets of particles form coalitions which will be represented by superordinate players.
In contrast to atomic splittable routing games, flow is modeled by a flow over time and players' decisions might be
adapted to new situations. This extension covers a greater variety of scenarios. For example, road traffic models in
which most drivers are guided by navigation systems (e.g. Google Maps, TomTom, Here, Garmin) can be modeled by covering the
strategic behavior of the firms: The decisions of a single driver do not have a big impact on the city's traffic, but
Google Maps decisions do; and TomTom might actually want to react. The situation will even intensify in the future with
the rise of autonomous driving, as the decision making process is shifted to the navigation systems. We would like to
point out that -- in contrast to previous models -- the interests of navigation companies and the general public are
assumably in line: The cooperation of the users of a navigation system could reduce the average driving time per company
on the one hand and the driving time in general on the other hand. This would lead to lower energy consumption, and
therefore, lower emissions of polluting substances.

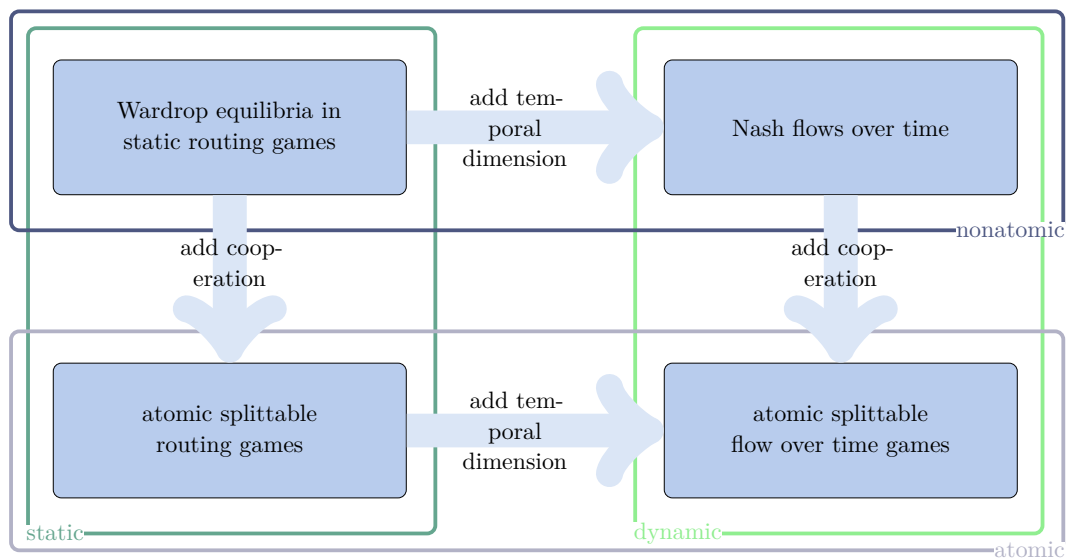
\begin{figure}[t]
    \centering
    \resizebox{\textwidth}{!}{
\begin{tikzpicture}[
roundrectangle/.style={rectangle, draw, fill=mynodecolor, rounded corners=.8ex, minimum width = 5.2cm, minimum height = 2cm, text width=4cm, align=center},
arrowstyle/.style={draw=mynodecolor!50, line width = 0.5cm},
arrownodestyle/.style={text width = 2cm, align=center, midway},
setstyle/.style={ultra thick, rectangle, draw, rounded corners=.8ex, fill opacity=0.1},
setannotation/.style={inner sep = 0pt, fill=white, shift={(0,0.04)}}
]

%% SET LIMITATIONS
\node[setstyle, draw=fillcolor1, anchor=north west, minimum height = 7.5cm, minimum width = 6cm] (static set) at (-8,3.5) {};
\node[setannotation, text=fillcolor1, anchor= west] (static) at (static set.south west) {static};

\node[setstyle, draw=fillcolor2, anchor=north east, minimum height = 7.5cm, minimum width = 6cm ] (dynamic set) at (7,3.5) {};
\node[setannotation, text=fillcolor2, anchor=west] (dynamic) at (dynamic set.south west) {dynamic};

\node[setstyle, draw=fillcolor3,anchor=north west, minimum height = 3.25cm, minimum width = 15.5cm] (nonatomic set) at (-8.25,3.75) {};
\node[setannotation, text=fillcolor3, anchor=east] (nonatomic) at (nonatomic set.south east) {nonatomic};

\node[setstyle, draw=fillcolor4, anchor=north west, minimum height = 3.25cm, minimum width = 15.5cm] (atomic set) at (-8.25,-1) {};
\node[setannotation, text=fillcolor4, anchor=east] (atomic) at (atomic set.south east) {atomic};

%% TEXT BOXES
\node[roundrectangle] (problem1) at (-5,2) {Wardrop equilibria in static routing games};
\node[roundrectangle] (problem2) at (4,2) {Nash flows over time};
\node[roundrectangle] (problem3) at (-5,-2.5) {atomic splittable routing games};
\node[roundrectangle] (problem4) at (4,-2.5) {atomic splittable flow over time games};

%% ARROWS
\draw[arrowstyle, ->] (problem1.south) -- (problem3.north) node [arrownodestyle, shift={(0,0.25)}] {add cooperation};
\draw[arrowstyle, ->] (problem1.east) -- (problem2.west) node [arrownodestyle, shift={(-0.3,0)}] {add temporal dimension};
\draw[arrowstyle, ->] (problem3.east) -- (problem4.west) node [arrownodestyle, shift={(-0.3,0)}] {add temporal dimension};
\draw[arrowstyle, ->] (problem2.south) -- (problem4.north) node [arrownodestyle, shift={(0,0.25)}] {add cooperation};

\end{tikzpicture}}
    \caption{Relationship between equilibrium situations in static routing games and atomic splittable flow over time games.}
    \label{fig:relationship_between_models}
\end{figure}

It turned out to be surprisingly challenging to consider equilibria in atomic splittable flow over time games. For this
reason the overall goal is to define a solid model on these dynamic games and to present some preliminary observations,
as well as some non-trivial first results, which serve as a basis for further research.

\paragraph{Related work.} Static network flows have been studied for quite a while. A lot of pioneer work is due to
Ford and Fulkerson, who also were the first to introduce \emph{flows over
time}~\cite{ford1958constructing,ford1962flows}. They provided an efficient algorithm for a \emph{maximum flow over
time}, which sends the maximal flow volume from a source to a sink given a finite time horizon. Closely related, a
\emph{quickest flow} minimizes the arrival time of the latest particle for a given flow volume. This can be achieved by
combining the algorithm of Ford and Fulkerson with a binary search
framework~\cite{burkard1993quickest,Fleischer_Tardos_1998}. For single-source and single-sink networks it is furthermore
possible to construct a flow over time that is maximal for all time horizons (and quickest for all flow volumes)
simultaneously. The existence of these so-called \emph{earliest arrival flows} was shown by Gale in
1959~\cite{Gale_1959}. They can be computed algorithmically by using the successive shortest path algorithm in the
residual networks~\cite{minieka1973}. For more details and further references to literature on optimization problems in
the flow over time setting, we refer to the survey of Skutella~\cite{skutella2009introduction}.

Koch and Skutella \cite{Koch_Skutella_2011} approach flows over time from a game theoretic perspective
by introducing \emph{Nash flows over time}. In their model, every infinitesimally small flow particle is considered to
be a player aiming to reach the common destination as early as possible. As the flow rate entering an edge could exceed
its capacity, they considered the \emph{deterministic queuing model}~\cite{vickrey1969congestion}, which causes the
excess flow to wait in a queue in front of the bottle-neck. Existence of these \emph{dynamic equilibria}
were shown by Cominetti et al.~\cite{cominetti2011existence}. Several other aspects, including uniqueness, continuity, long term behavior,
multi-terminals, spillback and price of anarchy, were studied in recent years
\cite{bhaskar2015stackelberg,cominetti2021long,Cominetti_Correa_Larre_2015,Correa_Cristi_Oosterwijk_2019,israel2020impact,olver2022continuity,pham2020dynamic,sering2018multiterminal,sering2019nash}; see \cite{sering2020nash} for an overview. 
A slightly different approach for user equilibria was presented by Graf et al.~\cite{graf2020price,graf2022finite,graf2020ide}. They use the same
flow over time model, except that particles do not anticipate the future evolution of the flow, but instead choose
quickest routes according to current waiting times. As these delays may be subject to change, each particle can adapt
its route choice along the way.

Atomic splittable congestion games for static network flows can be described as Wardrop equilibria~\cite{Wardrop_1952}
with coalitions~\cite{Haurie_Marcotte_1985,marcotte1987algorithms}; see also the survey of Correa and Stier-Moses
\cite{correa2010wardrop}. For these games, Nash equilibria always exist, which can be shown by standard fixed point
techniques~\cite{rosen1965existence}. Altman et al.~\cite{Altman_Basar_Jimenez_Shimkin_2002} showed that equilibria are
unique if the delay functions are polynomials of degree less than $3$. Regarding more general delay functions, Bhaskar
et al.~\cite{Bhaskar_2015} showed that for two players a unique equilibrium exists if, and only if, the network is a
\emph{generalized series-parallel graph}. Harks and Timmermans \cite{harks2018uniqueness} showed uniqueness of
equilibria when the players' strategy space has a bidirectional flow polymatroid structure.
Roughgarden~\cite{Roughgarden_2005} showed that the inefficiency of a system decreases with an increasing degree of
cooperation. He showed that the price of anarchy for classes of traveling time functions in the atomic case is bounded
by the price of anarchy for the same class of functions in the nonatomic case. Further research on the price of anarchy
in static atomic splittable games is due to Cominetti et al.~\cite{cominetti2009impact},
Harks~\cite{harks2011stackelberg}, and Roughgarden and Schoppmann~\cite{roughgarden2015local}. Computational-wise
Cominetti et al.~\cite{cominetti2009impact} showed that equilibria can be computed efficiently when the cost functions
are affine and player-independent. Regarding player-specific affine costs, Harks and
Timmermans~\cite{harks2017equilibrium} described a polynomial algorithm for parallel-edge networks, and Bhaskar and
Lolakapuri \cite{bhaskar2018equilibrium} presented an exponential algorithm for general convex functions. Very recently,
Klimm and Warode showed that the computation with player-specific affine costs is \PPAD-complete for general networks
\cite{klimm2020complexity}.

We should also mention that the combination of cooperation and temporal dimension has been considered for discrete packet
routing games; see Peis et al.~\cite{peis2018oligopolistic}. Here, each player controls a finite amount of packets, which
has to be routed through a network in discrete time steps. For more results on competitive packet routing models
we refer to Hoefer et al.~\cite{hoefer2009competitive} (continuous-time packets model) and Harks et al.~\cite{harks2018competitive} (discrete-time packet model).

\paragraph{Contribution and overview.} In \Cref{sec:model}, we introduce all notations and formally describe atomic
splittable flow over time games for general networks. The players' strategies determine how much flow they assign to
each edge for every point in time during the game depending on available information on the current state. We consider
two very natural sets of information. The first consists solely of the current time. In \Cref{sec:naive_strategy_space}
we show that this setting does not allow for a Nash equilibrium in general, not even in a network with only two parallel
edges. This motivates to consider more complex information models. Hence, \Cref{sec:strategy_space_with_exit_times} is
dedicated to the second set of information which additionally comprises the current congestion of the edges in form of
the exit times. As the first main result we show that every strategy profile results in a unique feasible flow over time
by formulating the conditions as initial value problem and applying the Picard-Lindel\"of theorem. For parallel-edge
networks we show that Nash equilibria always exist by explicitly stating a strategy profile. Furthermore, we prove for that setting that all Nash equilibria have the same objective equal to the system optimum (i.e., the value of a maximum flow over time) implying that the price of anarchy for those networks is~$1$. Finally, we suggest
further areas of research in \Cref{sec:final}.

\section{Atomic Splittable Flow Over Time Games} \label{sec:model} 
In this section, we are going to properly define
atomic splittable flow over time games. The two main aspects are the multi-commodity flow dynamics
(see \cite{sering2018multiterminal}) and the players' strategies, which depend on the information received over time.

\paragraph{Game setting.} A \emph{network} consists of a directed graph $G = (V, E)$, where every edge $e \in E$ is
equipped with a \emph{transit time} $\tau_e > 0$ and a \emph{capacity} $\nu_e >0$. For a node $v \in V$ we denote the
set of all incoming edges by $\delta_v^-$ and the set of outgoing edges by $\delta_v^+$.

For an \emph{atomic splittable flow over time game} we consider a finite set of players $P$, each with an
origin-destination pair $s_j$-$t_j$ and a \emph{supply rate} $d_j >0$, as well as a time horizon~$H > 0$. We assume that
$s_j$ can reach $t_j$ within the network.

The flow of player $j$ enters the network via node $s_j$ at a rate of $d_j$ from time~$0$ onwards. The goal is to
maximize the cumulative flow volume reaching node $t_j$ before the end of the game at time $H$.

%%%%% FLOWS %%%%%%%%%%%%%%%%%%%%%%%%%%%%%%%%%%%%%%%%%%%%%%%%%%%%%%%%%%
\paragraph{Flow dynamics.} In the deterministic queuing model the total inflow rate into an edge is bounded by the
capacity. If the capacity is exceeded, a queue builds up in which all entering particles have to wait in line.
Afterwards each particle needs $\tau_e$ time to traverse the edge before it can enter the next edge along the path; see
\Cref{fig:queue_and_edge}.  The dynamics of this process are formally defined as follows:

A \emph{flow over time} is described by a family of Lebesgue-integrable functions $f = (f^+, f^-) = (f_{e,j}^+,
f_{e,j}^-)_{e \in E, j \in P}$, where $f_{e,j}^{+}, f_{e,j}^{-}: [0, H) \to \mathbb{R}_{\geq 0}$ denote the rate at
which flow controlled by player $j$ enters and leaves edge $e$. The flow rates summed over all players $\sum_{j \in P}
f_{e,j}^+ (\theta)$ and $\sum_{j \in P} f_{e,j}^- (\theta)$ are called \emph{total in- and outflow rates}. The
\emph{cumulative in- and outflows}, i.e., the amount of player $j$'s flow that has entered or left an edge $e$ up to
time $\theta$, is denoted by $F_{e,j}^+ (\theta) = \int_{0}^{\theta} f_{e,j}^+ (\vartheta) \diff \vartheta$ and
$F_{e,j}^- (\theta) = \int_{0}^{\theta} f_{e,j}^- (\vartheta) \diff \vartheta$. Finally, $f^+$, $f^-$, $F^+$ and $F^-$
denote the vectors of (cumulative) in- and outflows with one entry per edge-player-pair.

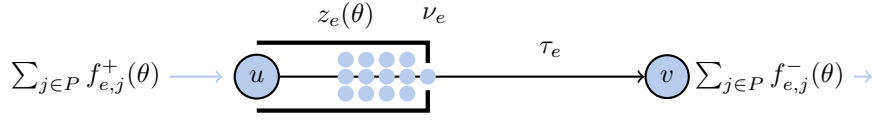
\begin{figure}[t]
    \centering
    \begin{tikzpicture}[scale=0.9,
roundnode/.style={circle, thick, draw=black, fill=mynodecolor},
particle/.style={circle, fill=mynodecolor, minimum size=6pt, inner sep=0pt}
]

%TODO: eventuell noch (\theta) bei den Flüssen einfügen

%% NODES 
\node[roundnode] (tail) at (0,0) {$u$};
\node[roundnode] (head) at (6,0) {$v$};

%% LABELS
\node (inflow) at (-2.5,0) {$\sum_{j \in P} f_{e,j}^+(\theta)$};
\node (outflow) at (7.5,0) {$\sum_{j \in P} f_{e,j}^-(\theta)$};
\node (queue) at (1.3, 0.9) {$z_e(\theta)$};
\node (service rate) at (2.6, 0.9) {$\nu_e$};
\node (transit time) at (4.3, 0.4) {$\tau_e$};
\node (blank) at (3,-1) {};

%% LINES
\draw[->, thick] (tail.east) -- (head.west);
\draw[->, thick, mynodecolor] (inflow.east) -- (-0.5,0);
\draw[->, thick, mynodecolor] (outflow.east) -- (9,0);
\draw[ultra thick] (0,0.5) -- (2.5,0.5)-- (2.5,0.2);
\draw[ultra thick] (0,-0.5) -- (2.5,-0.5)-- (2.5,-0.2);

%% PARTICLES

\node[particle] at (2.5,0) {}; %

\node[particle] at (2.2,0.25) {};%
\node[particle] at (2.2,0) {};%
\node[particle] at (2.2,-0.25) {};%

\node[particle] at (1.9,0.25) {};%
\node[particle] at (1.9,0) {};%
\node[particle] at (1.9,-0.25) {};%

\node[particle] at (1.6,0.25) {};%
\node[particle] at (1.6,0) {};%
\node[particle] at (1.6,-0.25) {};%

\node[particle] at (1.3,0.25) {};%
\node[particle] at (1.3,0) {};%
\node[particle] at (1.3,-0.25) {};%

%\begin{comment}
%%% LABELS
%\node (inflow 1) at (-3.5, .5)  {$d_1 = 4$};
%\node (inflow 2) at (-3.5, -.5) {$d_2 = 4$};
%\node (capacity 1) at (0,1)  {$\nu_1 = 2$};
%\node (capacity 2) at (0,-1) {$\nu_2 = 2$};
%\node (edge 1) at (2,1)  {$e_1$};
%\node (edge 2) at (2,-1) {$e_2$};
%\node (transit time 1) at (4,1)  {$\tau_1 = 0$};
%\node (transit time 2) at (4,-1) {$\tau_2 = 0$};
%\node (time horizon) at (2, 2) {$H = 20$};
%
%
%
%%% LINES
%\draw[->, thick] (source.east) .. controls (2,.5)  .. (sink.west);
%\draw[->, thick] (source.east) .. controls (2,-.5) .. (sink.west);
%\draw[->, thick] (inflow 1) -- (-1.5, .5);
%\draw[->, thick] (inflow 2) -- (-1.5, -.5);
%
%%% SPACE
%\draw[draw=white] (sink) -- (7,0);
%\end{comment}

\end{tikzpicture}
    \caption[Representation of an edge in the deterministic queuing model]{Representation of an edge in the deterministic
    queuing model: If more flow particles enter edge $e=uv$ within the total inflow rate $\sum_{j \in P} f^+_{e,j}(\theta)$
    than its capacity $\nu_e$ allows to process, they build up a queue, whose current length is given by $z_e(\theta)$.
    Whenever the queue is non-empty at time $\theta$ the total outflow rate $\sum_{j \in P} f^-_{e,j}(\theta + \tau_e)$ at time $\theta + \tau_e$ equals the capacity $\nu_e$.}
    \label{fig:queue_and_edge}
\end{figure}

 Such a family of functions $f$ is a \emph{flow over time} if the following two conditions hold for all
 $e \in E$ and $j \in P$:
 
\emph{Flow conservation} is fulfilled for all $\theta \in [0, H)$:
\begin{equation}\begin{aligned}
\sum\limits_{e \in \delta^+_v} f_{e,j}^{+}(\theta) -
\sum\limits_{e \in \delta^-_v} f_{e,j}^{-}(\theta) =
\begin{cases}
    d_j     &\text{for } v = s_j,\\
    0       &\text{for } v \in V \setminus \set{ s_j,t_j }.
\end{cases}\label{eq:flow_conservation}
\end{aligned}\end{equation}
\emph{Non-deficit constraints} are satisfied for all $\theta \in [0, H - \tau_e)$:
    \begin{equation}
        F_{e,j}^+ (\theta) - F_{e, j}^- (\theta + \tau_e) \geq 0. \label{eq:exit_time_derivative}
    \end{equation}

To track the net flow that is not yet processed and remains in the \emph{queue}, we introduce $z_e(\theta)$ to denote
the \emph{queue length} at time $\theta$. Formally, it is defined as $z_e(\theta) = \sum_{j \in P} F_{e,j}^+(\theta) - \sum_{j \in P} F_{e,
j}^-(\theta + \tau_e)$ for $e \in E$. For a \emph{feasible} flow over time we require that, whenever flow is
waiting in the queue, the edge operates at capacity rate. In other words, for all $\theta \in [0, H - \tau_e)$ and $ e \in
E$, we require
\begin{equation}\begin{aligned} \label{eq:operating_a_capacity}
    \sum_{j \in P} f_{e,j}^-(\theta + \tau_e) = 
    \begin{cases}
        \nu_e       & \text{if } z_e(\theta) > 0, \\
        \min\set{ \sum_{j \in P} f_{e, j}^+(\theta),\, \nu_e }   &\text{else.}
    \end{cases}
\end{aligned}\end{equation}

 The \emph{waiting time} $q_e(\theta)$ experienced by a particle entering edge~$e$ at time~$\theta$ is generally defined
 as the time needed to process the flow present in the queue when the particle enters it. In other words, it is the time
 span between its entrance to and its exit from the queue just before traversing the edge. That is
\[ q_e(\theta) \coloneqq \min \Set{q \geq 0 | \int_{\theta}^{\theta + q} \sum_{j \in P} f_{e, j}^{-} (\vartheta + \tau_e) \diff \vartheta = z_e(\theta)} = \frac{z_e(\theta)}{\nu_e}.\]

The \emph{exit time} $T_e(\theta)$ of a particle entering an edge $e \in E$ at time $\theta$ is given by the sum of the
entrance time $\theta$, its waiting time in the queue $q_e(\theta)$ and the transit time $\tau_e$. Hence, we have
\[T_e(\theta) \coloneqq \theta + \frac{z_e(\theta)}{\nu_e} + \tau_e.\] 
Since $z_e$ can at most decrease at rate $\nu_e$, it
follows that $q'_e(\theta) \geq -1$ and $T'_e(\theta) \geq 0$, which induces that $T_e$ is non-decreasing. Note that
these derivatives exist for almost all $\theta$ due to Lebesgue's differentiation theorem.

Finally, in a feasible flow over time, the flow of different players should merge seamlessly. This means that, at any point in time, a player's
share of the total outflow rate is equal to her share of the total inflow rate back at the time when the flow entered the edge. More precisely, we
require
\begin{equation}
    f_{e,j}^-(\theta) = 
    \begin{cases}
        \frac{f_{e,j}^+(\phi)}{\sum_{j' \in P} f_{e, j'}^+(\phi)} \cdot \sum_{j' \in P} f_{e, j'}^-(\theta) &\text{if } \sum_{j' \in P} f_{e, j'}^+(\phi) > 0, \\
        0 &\text{else,}
    \end{cases} \label{eq:individual_outflows}
\end{equation}
for all $\theta \in \horizon$ and $\phi \coloneqq \max T_e^{-1}(\theta)$. Here, we set $f_{e, j'}^+(\phi) \coloneqq 0$ for $\phi < 0$. Note that $\phi$ denotes the edge-entrance times of all particles leaving the edge at time $\theta$.
Taking the maximum of $T_e^{-1}(\theta)$ simply ensures well-definedness which is required since $T_e$ might not be strictly increasing. 

To conclude we denote the set of all \emph{feasible flows over time} by $\mathcal{F}$, i.e., all $f = (f^+, f^-)$ that
satisfy \eqref{eq:flow_conservation}, \eqref{eq:exit_time_derivative}, \eqref{eq:operating_a_capacity} and
\eqref{eq:individual_outflows}. As the outflow rates are uniquely defined by the inflow rates, we refer to a feasible
flow over time only by the corresponding inflow $f^+$, and write $f^+ \in \mathcal{F}$.

\paragraph{Atomic splittable flow over time games.} 
%\subsubsection{Payoff Functions.} 
Let $\rho_j: \mathcal{F} \to \mathbb{R}$ be the function indicating player $j$'s
\emph{payoff} for a given $f \in \mathcal{F}$, which is to be maximized. In general, $\rho$ can be set to various
objective functions (e.g. arrival time of the player's latest particle or the average arrival time), but in this paper we will focus on the
\emph{maximum flow over time problem}.
Each player wants to maximize her amount of flow routed from $s_j$ to
$t_j$ before the end of the game at time $H$:
\[\rho_j(f) = \sum_{e \in \delta^-_{t_j}} F_{e,j}^{-}(H) - \sum_{e \in \delta^+_{t_j}} F_{e,j}^{+}(H).\]
We choose this maximum flow objective as it seems to be the most straight-forward payoff-function. It is conceivable
though, that most results might transfer to quickest flow payoff-functions via a binary-search framework (cf.\ in
non-competitive settings quickest flows are constructed from maximum flows over time via binary-search). But we leave
this for future research.

%\subsubsection{Strategies.} 
The \emph{strategy space} is a player's set of viable options in order to maximize
her payoff. A single \emph{strategy} is a complete instruction determining the player's
inflow rates of all times and for all situations possibly occurring.
Formally, the strategy space of player $j$ is a set of functions 
\[\StratSpace_j = \Set{g_j: \mathcal{I} \to [0, 1]^E | \begin{array}{l}
    g_{e,j} \text{ is Lebesgue-integrable for all } e \in E \text{ and }\\[2pt]
    \sum_{e \in \delta^+_{v}} g_{e,j}(I) = 1 \text{ for all } I \in \mathcal{I}, v \in V \setminus\set{t_j}
    \end{array}   },\]
where $\mathcal{I}$ is the set of information available to the players. This set is not well-defined yet, but we will discuss this extensively, and in the end, we will consider two separate definitions for $\mathcal{I}$, one in \Cref{sec:naive_strategy_space} and one in \Cref{sec:strategy_space_with_exit_times}.
Informally, the set of information is used to delineate what defines a situation and how it is perceived by the players. The interpretation is as
follows. For every information $I \in \mathcal{I}$, the value $g_{e,j}(I)$ determines which proportion of player $j$'s
flow arriving at $v$ is distributed onto the outgoing edges $e \in \delta_v^+$. We use proportions that sum up to $1$
instead of the inflow rates, as this easily ensures that flow conservation is fulfilled at all times. Note that the
received information $I$ can depend on the current time $\theta$ and on the flow over time~$f$ itself. As we normally do
not want the players to see the future, it should only depend on the flow over time up to time $\theta$. In general
it might be player dependent, hence, we write $I_j(\theta, f)$.

%\subsubsection{Transformator.}
In order to turn a \emph{strategy profile} $g =
(g_j)_{j \in P} \in \bigtimes_{j \in P} \Sigma_j$ into a feasible flow over time $f$, we consider the following system of equations that has to be satisfied for all $\theta \in [0, H)$:
\begin{equation} \label{eq:transformation_equation}
\begin{aligned}
f_{e, j}^+(\theta) &= g_{e, j}(I_j(\theta, f)) \cdot \bigg(\sum_{e' \in \delta^-_{s_j}} f_{e', j}^-(\theta) + d_j\bigg)  &&\text{ for }  j \in P, e \in \delta_{s_j}^+,\\
f_{e, j}^+(\theta) &= g_{e, j}(I_j(\theta, f)) \cdot \sum_{e' \in \delta^-_{u}} f_{e', j}^-(\theta) && \hspace{-1.5cm}\text{ for }  j \in P, e = uv \in E \setminus (\delta_{s_j}^+ \cup \delta_{t_j}^+),\\
f_{e, j}^+(\theta) &= 0  &&\hspace{-1.5cm}\text{ for }  j \in P, e \in \delta_{t_j}^+.
\end{aligned}
\end{equation}
Note that in order to keep it as simple as possible, we assume that flow of player~$j$ reaching $t_j$ leaves the network
immediately, hence, $f_{e,j}^+(\theta) = 0$ for all $e \in \delta_{t_j}^+$. This leads to a simpler payoff for all players $j$ of
$\rho_j(f) = \sum_{e \in \delta^-_{t_j}} F_{e,j}^{-}(H)$.

Since the inflow rates $f^+(\theta)$ might depend on the flow over time itself up to time~$\theta$, it is not guaranteed that
this system of equations yields a feasible flow over time as solution. 

To illustrate this issue assume for example a game with a single player in a network with two parallel edges $e_1$ and $e_2$ where $e_1$ has a tiny capacity. If the player's strategy is to send everything into $e_1$ as long as there is no waiting time on $e_1$, and otherwise send everything into $e_2$, this would not result in a feasible flow over time. To see this, suppose that the inflow into $e_1$ would be positive on a measurable set, which would immediately cause a positive waiting time on $e_1$. Hence, the strategy says that no flow is sent into $e_1$. On the other hand, if no flow is sent into $e_1$ at all there would not be any waiting time, leading again to a contradiction.

For that reason, the key challenge is to find a
reasonable set of information $\mathcal{I}$ and strategy spaces $\Sigma_j$, such that, on the one hand, there exists a unique
(up to a null set) feasible flow over time satisfying \eqref{eq:transformation_equation} for every given strategy
profile~$g$, and on the other hand, a Nash equilibrium exists. Note that we only consider \emph{pure} Nash equilibria.

In the following two sections we discuss two very natural sets of information.

\section{Temporal Information Only} \label{sec:naive_strategy_space} First, we want to examine the simplest set of
information, namely the current point in time only: We set $I_j(\theta, f) \coloneqq \theta$ for all players $j \in P$.
That means the players do not receive any information about the current state of the flow, but instead have to decide at
the beginning of the game along which routes their flow particles are routed. In this case it is guaranteed that there
exists a unique feasible flow over time that satisfies \eqref{eq:transformation_equation}, as we can simply set $f_{e,
j}^+(\theta)$ to the right sides of \eqref{eq:transformation_equation} (formally this also follows from
\Cref{thm:flow_depends_on_time_and_exit_times}). However, Nash equilibria do not exist in general, which is already true
for very simple networks.

\begin{theorem}\label{thm:no_equilibrium}
With temporal information only, there exists no Nash equilibrium in an atomic splittable flow over time game of two
players $p_1$ and $p_2$ with identical supply rates $d_{1}, d_{2}\!>\!0$, on a network with two identical parallel edges
$e_1, e_2$ from $\source\;$( $\!=\!\source_{p_1}\!=\!\source_{p_2})$ to $\sink\;$( $\!=\!\sink_{p_1}\!=\!\sink_{p_2})$
with $\nu_{e_1} \!=\! \nu_{e_2} < d_{1} \!=\! d_{2}$ and $\tau_{e_1} \!=\! \tau_{e_2} < H$.
\end{theorem}

The key proof idea is the following. To every strategy of the competitor, a player can choose a response strategy that
yields a payoff of strictly more than half the total optimum, i.e., the flow value of a maximum flow over time with
inflow rate $d_1 + d_2$. This can be achieved by first \emph{mirroring} the competitor's strategy (copying the strategy
but interchanging the roles of $e_1$ and $e_2$) and then shifting some flow in the beginning of the game. This shift
changes the points in time $r_i$ when the last particles arriving at $t$ just in time, enter edge $e_i$ for $i = 1,2$.
As the new values $\hat r_1$ and $\hat r_2$ are not equal anymore the responding player can modify the inflow rates
between $\hat r_1$ and $\hat r_2$ in order to squeeze in a little more flow than the competitor, which will then be a little more than $\Opt/2$.

As in turn the competitor can again choose a strategy with a payoff
of more than half the optimum, this immediately implies that a Nash equilibrium cannot exist. This idea for the network
given in \Cref{fig:standard_graph} is visualized in \Cref{fig:mirror_strategy,fig:winning_strategy}.
% and the formal proof can be found in the appendix on page~\pageref{proof:no_equilibrium}.

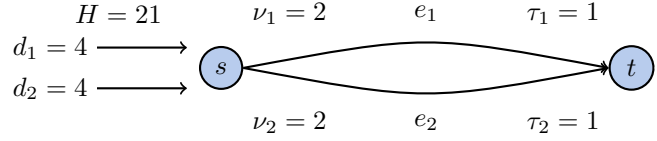
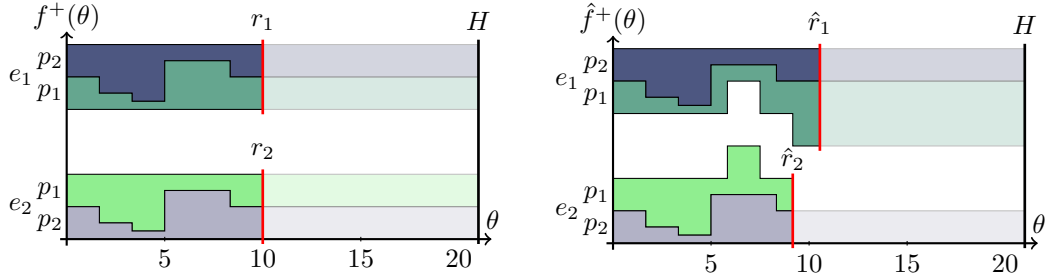
\begin{figure}[t]
\begin{subfigure}{\linewidth}
  \centering
  \begin{tikzpicture}[scale=0.9,
roundnode/.style={circle, thick, draw=black, fill=mynodecolor}
]
%% LABELS
\node (inflow 1) at (-3.5, .3)  {$d_1 = 4$};
\node (inflow 2) at (-3.5, -.3) {$d_2 = 4$};
\node (capacity 1) at (0,0.8)  {$\nu_1 = 2$};
\node (capacity 2) at (0,-0.8) {$\nu_2 = 2$};
\node (edge 1) at (2,0.8)  {$e_1$};
\node (edge 2) at (2,-0.8) {$e_2$};
\node (transit time 1) at (4,0.8)  {$\tau_1 = 1$};
\node (transit time 2) at (4,-0.8) {$\tau_2 = 1$};
\node (time horizon) at (-2.5, 0.8) {$H = 21$};

%% NODES 
\node[roundnode] (source)   at (-1,0) {$s$};
\node[roundnode] (sink)     at (5,0) {$t$};

%% LINES
\draw[->, thick] (source.east) .. controls (2,.5)  .. (sink.west);
\draw[->, thick] (source.east) .. controls (2,-.5) .. (sink.west);
\draw[->, thick] (inflow 1) -- (-1.5, .3);
\draw[->, thick] (inflow 2) -- (-1.5, -.3);

%% SPACE
\draw[draw=white] (sink) -- (7,0);

\end{tikzpicture}
  \caption{A parallel network with two identical edges and two identical players.}
  \label{fig:standard_graph}
\end{subfigure}
  \tikzset{every picture/.style={scale=0.43}}%
  \begin{subfigure}{0.48\linewidth}  
    \centering
    % Quelle: http://www.texample.net/tikz/examples/parabola-plot/
\begin{tikzpicture}
    
    \draw[draw=black, fill=fillcolor3] (0,6) -- (6,6) -- (6,5) -- (5,5) -- (5,5.5) -- (3,5.5) -- (3,4.25) -- (2,4.25) -- (2,4.5) -- (1,4.5) -- (1,5) -- (0,5) -- (0,6);
    \draw[draw=black, fill=fillcolor3, nearly transparent] (6,5) rectangle ++(6.6,1);
    
    \draw[draw=black, fill=fillcolor1] (0,4) -- (6,4) -- (6,5) -- (5,5) -- (5,5.5) -- (3,5.5) -- (3,4.25) -- (2,4.25) -- (2,4.5) -- (1,4.5) -- (1,5) -- (0,5) -- (0,4);
    \draw[draw=black, fill=fillcolor1, nearly transparent] (6,4) rectangle ++(6.6,1);
    
    \draw[draw=black, fill=fillcolor2] (0,2) -- (6,2) -- (6,1) -- (5,1) -- (5,1.5) -- (3,1.5) -- (3,0.25) -- (2,0.25) -- (2,0.5) -- (1,0.5) -- (1,1) -- (0,1) -- (0,2);
    \draw[draw=black, fill=fillcolor2, nearly transparent] (6,1) rectangle ++(6.6,1);
    
    \draw[draw=black, fill=fillcolor4] (0,0) -- (6,0) -- (6,1) -- (5,1) -- (5,1.5) -- (3,1.5) -- (3,0.25) -- (2,0.25) -- (2,0.5) -- (1,0.5) -- (1,1) -- (0,1) -- (0,0);
    \draw[draw=black, fill=fillcolor4, nearly transparent] (6,0) rectangle ++(6.6,1);
    
    \draw[->, thick] (-0.2,0) -- (13,0) node[right, shift={(-0.2, 0.26)}] {$\theta$};
    \draw[->, thick] (0,-0.2) -- (0,6.4) node[above, shift={(0, -0.12)}] {$f^+(\theta)$};
    
    \foreach \y/\ytext in {0.5/p_2, 1.5/p_1, 4.5/p_1, 5.5/p_2}
    \node at (-0.5, \y) {$\ytext$};
    
    \foreach \edge/\edgetext in {1/e_2, 5/e_1}
    \node at (-1.4, \edge) {$\edgetext$};
    
    \draw[draw=red, line width=1pt] (6,-0.15) -- (6,2.15) node[above] {$r_2$};
    \draw[draw=red, line width=1pt] (6,3.85) -- (6,6.15) node[above] {$r_1$};
    
    \draw[draw=black, line width = 1pt] (12.6,-0.15) -- (12.6,6.15) node[above] {$H$};

    \foreach \x/\xtext in {3/5, 6/10, 9/15, 12/20}
    \draw[shift={(\x,0)}, draw=none] (0pt,2pt) -- (0pt,-2pt) node[below] {$\xtext$};
    \foreach \x/\xtext in {3/6, 9/18}
    \draw[shift={(\x,0)}] (0pt,2pt) -- (0pt,-2pt);

\end{tikzpicture}
    \caption{Step 1: $p_1$ mirrors $p_2$'s strategy. We have $r_1 = r_2 = 10$ and both queues build up from $\theta = 0$
    onward. W.l.o.g., we have $f^+_{1,2}(\theta) \leq d_2 = 4$ for $\theta \in [10 - \varepsilon, 10)$.}
    \label{fig:mirror_strategy}
  \end{subfigure} \quad
  \begin{subfigure}{0.48\linewidth}
    \centering 
    % Quelle: http://www.texample.net/tikz/examples/parabola-plot/
\begin{tikzpicture}
    
    \draw[draw=black, fill=fillcolor3] (0,6) -- (6.33,6) -- (6.33,5) -- (5,5) -- (5,5.5) -- (3,5.5) -- (3,4.25) -- (2,4.25) -- (2,4.5) -- (1,4.5) -- (1,5) -- (0,5) -- (0,6);
    \draw[draw=black, fill=fillcolor3, nearly transparent] (6.33,5) rectangle ++(6.27,1);
    
    \draw[draw=black, fill=fillcolor1] (0,4) -- (3.5,4) -- (3.5,5) -- (4.5,5) -- (4.5,4) -- (5.5,4) -- (5.5,3) -- (6.33,3) -- (6.33,5) -- (5,5) -- (5,5.5) -- (3,5.5) -- (3,4.25) -- (2,4.25) -- (2,4.5) -- (1,4.5) -- (1,5) -- (0,5) -- (0,4);
    \draw[draw=black, fill=fillcolor1, nearly transparent] (6.33,3) rectangle ++(6.27,2);
    
    \draw[draw=black, fill=fillcolor2] (0,2) -- (3.5,2) -- (3.5,3) -- (4.5,3) -- (4.5,2) -- (5.5,2) -- (5.5,1) -- (5,1) -- (5,1.5) -- (3,1.5) -- (3,0.25) -- (2,0.25) -- (2,0.5) -- (1,0.5) -- (1,1) -- (0,1) -- (0,2);
    
    \draw[draw=black, fill=fillcolor4] (0,0) -- (5.5,0) -- (5.5,1) -- (5,1) -- (5,1.5) -- (3,1.5) -- (3,0.25) -- (2,0.25) -- (2,0.5) -- (1,0.5) -- (1,1) -- (0,1) -- (0,0);
    \draw[draw=black, fill=fillcolor4, nearly transparent] (5.5,0) rectangle ++(7.1,1);
    
    \draw[->, thick] (-0.2,0) -- (13,0) node[right, shift={(-0.2, 0.26)}] {$\theta$};
    \draw[->, thick] (0,-0.2) -- (0,6.4) node[above, shift={(0, -0.12)}] {$\hat{f}^+(\theta)$};
    
    \foreach \y/\ytext in {0.5/p_2, 1.5/p_1, 4.5/p_1, 5.5/p_2}
    \node at (-0.5, \y) {$\ytext$};
    
    \foreach \edge/\edgetext in {1/e_2, 5/e_1}
    \node at (-1.4, \edge) {$\edgetext$};
    
    \draw[draw=red, line width=1pt] (6.33,2.85) -- (6.33,6.15) node[above] {$\hat{r}_1$};
    \draw[draw=red, line width=1pt] (5.5,-0.15) -- (5.5,2.15) node[above, shift={(0,-0.09)}] {$\hat{r}_2$};
    
    \draw[draw=black, line width = 1pt] (12.6,-0.15) -- (12.6,6.15) node[above] {$H$};

    \foreach \x/\xtext in {3/5, 6/10, 9/15, 12/20}
    \draw[shift={(\x,0)}, draw=none] (0pt,2pt) -- (0pt,-2pt) node[below] {$\xtext$};
    \foreach \x/\xtext in {3/5, 6/10, 9/15}
    \draw[shift={(\x,0)}] (0pt,2pt) -- (0pt,-2pt);

\end{tikzpicture} 
    \caption{Step 2: $p_1$ shifts $\delta$ flow units from $e_1$ to $e_2$ causing $\hat{r}_1 > 10$ and $\hat{r}_2 < 10$.
    By pumping at maximum rate into $e_1$ after $\hat{r}_2$ her payoff is $\hat\rho_1 > 40 = \Opt/2$.}
    \label{fig:winning_strategy}
  \end{subfigure}
  \caption[Construction of an always winning response strategy]{Construction of a response strategy that always yields
  strictly more than half of the total optimum in a network of two parallel edges. Here, $r_i$ and $\hat r_i$ denote the
  time when the last particles arriving at $t$ just in time, enter $e_i$.} \label{fig:counter_example}
\end{figure}

\begin{proof}[Proof of \Cref{thm:no_equilibrium}]\label{proof:no_equilibrium}
For the sake of simplicity we replace indexes $e_i$ by $i$ and $p_j$ by $j$. E.g.
$f_{1,2}$ denotes the inflow rate function of $p_2$ into~$e_1$. Furthermore, $r_1$ and  $r_2$ denote the
points in time when the last particles arriving at $t$ before $H$ enter $e_1$ and $e_2$, respectively. We will use a hat
to denote parameters that have altered with the change in strategy. For example, $\rho$ will denote the initial
and $\hat{\rho}$ the new payoff. The optimal value of a maximum flow is $\Opt = \nu_1 \cdot (H - \tau_1) +
\nu_2 \cdot (H - \tau_2)$. We show that to every strategy of player $p_2$, player $p_1$ can choose a response strategy
that yields a payoff of strictly more than half the optimum. As in turn $p_2$ can again choose a strategy with a payoff
of more than half the optimum, this immediately implies that a Nash equilibrium cannot exist.
To achieve this, player $p_1$ mirrors $p_2$'s strategy. By this we mean that $p_1$ has the same strategy as $p_2$, but
the roles of the edges are interchanged. \Cref{fig:mirror_strategy} illustrates this behavior.
\[f_{1,1}^{+}(\theta) = f_{2,2}^{+}(\theta)  \qquad \text{ and } \qquad  f_{2,1}^{+}(\theta) = f_{1,2}^{+}(\theta)     \quad\text{ for } \theta \in \horizon.\]

Both total inflow rates into the edges are equal and exceed the capacities, i.e., $f_{1,1}^+(\theta) + f_{1,2}^+(\theta) =
f_{2,1}^+(\theta) + f_{2,2}^+(\theta) > \nu_1 = \nu_2$ for $\theta \in \horizon$. It is easy to see that both payoffs are $\Opt/2$. As capacities are exceeded, 
queues build up from the very beginning. At time $r_1 = r_2 = \Opt/(d_1 +
d_2)$ the last particles that will reach $t$ in time enter the network. The response strategy of $p_1$ now consists of
reallocating a little flow to benefit from deferred $\hat{r}_1$ and $\hat{r}_2$: W.l.o.g.\ let $p_2$ send not more than
half of her supply rate $d_2$ at any time in $[r_1 - \varepsilon, r_1)$ for $\varepsilon > 0$, into $e_1$. That means
player $p_1$ sends not less than half of $d_1$ into $e_1$. (If no such $\varepsilon$ exists, due to some crazy function behavior, one can consider averages instead.) 
Player $p_1$ reallocates some $\delta > 0$ flow units from $e_1$ to $e_2$. By doing so, $\hat{r}_2$ is shifted
to $r_2 - \frac{\delta}{\nu_2}$. Meanwhile we want that $\frac{\delta}{\nu_2} \leq \varepsilon$, that the $\delta$ flow
units are taken before time $r_2 - \varepsilon$ and that the queue of $e_1$ is never empty for $\theta > 0$. Since there
has to exist $p_1$-flow of positive measure up to time $r_1$ on $e_1$ (in case the only flow $p_1$ sends into $e_1$
is within $[r_1 - \varepsilon, r_1)$ we can choose a smaller $\varepsilon$), we can clearly set $\delta$ small enough to
fulfill these conditions. As a result, $\hat{r}_1$ is postponed, i.e., $\hat{r}_1 > \hat{r}_2$. Conclusively, every flow
sent until time $\hat{r}_2$ reaches $t$ before $H$. Hence, the payoffs of both players belonging to $[0,\hat{r}_2)$ are
equal. After time $\hat{r}_2$, player~$p_1$ can pump all her flow into $e_1$, since it is still eligible to
reach $t$ before $H$. See \Cref{fig:winning_strategy} for an illustration. Therefore,
\begin{align*}
\hat{f}_{1,1}^{+}(\theta) \coloneqq d_1 \begin{cases}
> d_2 / 2 \geq f_{1,2}^{+}(\theta)     &\theta  \in [\hat{r}_2, r_1) \\
\geq f_{1,2}^{+}(\theta)         &\theta  \geq r_1.
\end{cases} 
\end{align*}
Hence, $p_1$'s inflow rate into $e_1$ is strictly greater than $p_2$'s during a time period of positive measure, which
shows that the payoff of $p_1$ is strictly larger than that of $p_2$. As the sum of the payoffs equals $\Opt$ we have that $\hat\rho_1 > \Opt/2$.
\end{proof}

\section{Information on Exit Times} \label{sec:strategy_space_with_exit_times} The absence of
a Nash equilibrium for temporal information only was mainly due to the theoretical information advantage of the
deviating player. Player $p_1$ can respond to $p_2$'s strategy, while $p_2$ does not see $p_1$'s moves and
is unable to react. As it is very natural to require inter-player reactions over time,
we extend the information by the current congestion of the edges in form of the exit times~$T(\theta) \coloneqq
(T_e(\theta))_{e \in E}$. Formally, we define
\[I_j(\theta, f) \coloneqq (\theta,T(\theta)) \in \mathcal{I} \coloneqq \horizon \times \R_{\geq 0}^E \qquad 
\text{ for all } j \in P.\]
The reason for choosing exit times over waiting times or queue sizes, which both contain the same information about the
congestion in the networks as the exit times, is that the exit times are non-decreasing. %This will be useful later on.
As the exit time $T_e(\theta) = \theta + \tau_e + \frac{1}{\nu_e} \sum_{j \in P} \left(F_{e,j}^+(\theta) - F_{e,j}^-(\theta +
\tau_e)\right)$ depends directly on the cumulative flows the equations system \eqref{eq:transformation_equation} becomes a
system of differential equations. In order to show existence and uniqueness we use the Picard-Lindel\"of theorem. For
this reason, we require the strategies to be locally Lipschitz-continuous from the right in order to ensure uniqueness.
More formally, we say a strategy $g_{e,j}$ is \emph{right-Lipschitz} if for every $I \in \mathcal{I}$ there exists an $L
> 0$ such that there exists an $\varepsilon > 0$ with $\abs{g_{e,j} (I) - g_{e,j}(I + x)} \leq L \cdot \norm{x}$ for all
$x \in [0, \varepsilon] \times [0, \varepsilon]^E$.

\paragraph{Existence and uniqueness.} \label{sec:well-definedness_uniqueness}
First we show, that in this setting every right-Lipschitz strategy profile yields a unique feasible flow over time.
\begin{theorem} \label{thm:flow_depends_on_time_and_exit_times}
    For all right-Lipschitz strategy profiles $g = (g_j)_{j \in P}$ of an atomic splittable flow over time game with information on exit times, there exists a unique feasible flow over time $f^+$ satisfying \eqref{eq:transformation_equation}.
%    , i.e.,
%    \begin{equation} \label{eq:flow_depends_on_time_and_exit_times}
%    f_{e, j}^+(\theta) = g_{e, j}(\theta, T(\theta)) \cdot d_j \qquad \text{ for all } e\in E,\; j\in P,\; \theta \in \horizon. 
%    \end{equation}
%Hereby, uniqueness holds up to a null set. \leon{$d_j$ im Beweis einfuegen}
\end{theorem}

\begin{proof}
  We will construct the feasible flow over time satisfying~\eqref{eq:transformation_equation} step by step
  starting with the empty flow over time $f^+ \equiv 0$ up to time $0$. We define a \emph{restricted flow over
  time} on the interval $[0, a)$ to be a vector of Lebesgue-integrable functions $(f_{e, j}^+)_{e \in E, j \in P}$, such that all flow conditions hold for all times in $[0, a)$.
  
  Suppose we are given a unique restricted feasible flow over time satisfying~\eqref{eq:transformation_equation} on the interval $[0, a)$ for some $a \in [0,H]$. If $a=H$, we are done. It is possible to determine the queue
  lengths $z(a) = (z_e(a))_{e \in E}$, the waiting times~$q(a) = (q_e(a))_{e \in E}$ and the exit times $T(a) =
  (T_e(a))_{e \in E}$ based on $f^+|_{[0, a)}$. Hence, we can also evaluate $g(a, T(a))$. In order to further extend the
  flow over time, we specify an interval $[a,b) \subseteq \horizon$. The right endpoint $b > a$ has to be small enough
  to ensure that $g_{e,j}(\theta, T(\theta))$ is Lipschitz-continuous for $\theta \in [a,b)$ and for all $e \in E$ and $j \in P$.
  We can find such $b$ as all $g_{e,j}$ are right-Lipschitz, the exit time functions~$T_e$ are non-decreasing and continuous in $\theta$. 
%  Even in the worst case, it cannot grow faster than $1 + \frac{1}{\nu_e} \cdot \sum_{j \in P} d_j$, which can directly be seen by
%  \[T'_e(\theta) = 1 + \frac{z'_e(\theta)}{\nu_e} \leq 1 + \frac{\sum_{j \in P} f_{e, j}^+(\theta)}{\nu_e} \leq 1 + \frac{1}{\nu_e} \cdot \sum_{j \in P} d_j .\]
%  Since $g_{e,j}$ is right-Lipschitz, we can find some $b > a$ such that all $g_{e,j}(\theta, T(\theta))$ stay constant on $\theta \in [a, b)$.
  
%  Also, $b$ has to be small enough such that if $q_e(a) > 0$, the waiting time remains strictly positive on $[a,b)$.
%  Since $q_e$ is a continuous function we can find such an endpoint $b$.

We obtain the following initial value problem:
\begin{align*}
f_{e, j}^+(a) &= g_{e, j}(a, T(a)) \cdot C_u(a),\\
f_{e, j}^+(\theta) &= g_{e, j}(\theta, \big(\theta + \tau_{e'} + \frac{1}{\nu_{e'}} \sum_{j' \in P} F_{e',j'}^+(\theta) - F_{e',j'}^-(\theta + \tau_{e'})\big)_{e' \in E}) \cdot C_u(\theta),
\end{align*} 
for all $j \in P$ and $e = uv \in E$. Here, $C_u(\theta) \geq 0$ is the total inflow rate into node~$u$ except for $t$
where it is $0$ (as stated in \eqref{eq:transformation_equation}). Since the transit times are strictly positive $C_u \colon [a, b) \to \R$ is completely determined by the restricted
feasible flow over time on $[0, a)$ as long as $b - a < \tau_e$. In addition, $C_u$ is right-Lipschitz, therefore, we can choose $b$ small enough, such that $C_u$ is Lipschitz-continuous on $[a, b)$.
 Furthermore, if
$q_e(\theta) > 0$ we have that $F_{e',j'}^-(\theta + \tau_{e'})$ is also determined from the past as long as $b - a <
q_e(\theta)$ or in the case of $q_e(\theta) = 0$, we have that $F_{e',j'}^-(\theta + \tau_{e'}) = \sum_{j' \in P}
F_{e',j'}^+(\theta)$ as required by \eqref{eq:exit_time_derivative}.

Since the exit times $T$ depend Lipschitz-continuous on $F^+$, also the right-side depends Lipschitz-continuous on
$F^+$. Hence, the Picard-Lindel\"of theorem guarantees the existence of a unique solution. It is easy to see, that extending $f$ by this solution yields a restricted feasible flow over time on $[0, b)$. Flow conservation is fulfilled as
$\sum_{e \in \delta_v^+} g_{e, j}(\theta, T(\theta)) = 1$ and \eqref{eq:exit_time_derivative},
\eqref{eq:operating_a_capacity} and \eqref{eq:individual_outflows} are satisfied as we implicitly define the outflow
rates accordingly.

It remains to show that it is possible to repeat the process until we cover the whole interval $\horizon$. Let
$b_i$, $i = 1,2,3,\dots$, be the right endpoints during this extension process. The sequence might converge to
$\lim\limits_{i \to \infty} b_i = b_{\infty} < H$. Existence and uniqueness are thus provided on $[0, b_{\infty})$.
But this means, we can apply the extension process for $a = b_{\infty}$. As we can always continue this process from a
limit point, we can apply this transfinite induction to obtain a unique feasible flow over time on $\horizon$.
\end{proof}

This theorem shows that we obtain a well-defined atomic splittable flow over time game as long as we only consider
strategies that are not too wild. It is worth noting that right-Lipschitz functions can have infinitely many jumps and
that we cannot hope for much more general strategy-functions as argued in the following. Suppose we allow for strategies
that are not continuous from the right. We end up with the following problem: Consider a game with only one player~$p_1$
and a network consisting of two edges $e_1$, $e_2$ both from $s_{p_1}$ to $t_{p_1}$ and with $\nu_{e_1} < d_{p_1}$. The
strategy with $g_{e_1, p_1}(\theta, T(\theta)) = 1$ if $T_{e_1}(\theta) \leq \theta + \tau_{e_1}$ and $0$ otherwise,
means, that if there is no queue on $e_1$ the players sends all its flow into $e_1$ (which causes a queue to build up)
but if there is a positive queue she sends nothing (which causes any positive queue to decrease). This strategy is not
right-continuous in $T$ as $g_{e_1, 1}$ switches from $1$ to $0$ as soon as $T > \tau_{e_1} + \theta$, and clearly, it
cannot lead to a feasible flow over time.

\paragraph{Existence of Nash equilibria in parallel-edge networks.} \label{sec:existence_of_Nash_equilibrium}
Unfortunately, the task to show the existence of a Nash equilibrium in this setting turns out to be quite challenging.
For this reason we only show the existence of Nash equilibria for simple networks: For the remaining of this section we
consider parallel-edge networks with two nodes $s = s_j$ and $t = t_j$, $j \in P$. We obtain
\[\StratSpace_j = \Set{g_j: \mathcal{I} \to [0, d_j]^E | \begin{array}{l}
    g_{e,j} \text{ is right-Lipschitz for all } e \in E\\
     \text{and } \sum_{e \in E} g_{e,j}(I) = 1 \text{ for all } I \in \mathcal{I} \end{array}}.\]

This leads to an ``over time'' version of \emph{atomic splittable singleton games} as they were studied in by Harks and
Timmermans~\cite{harks2017equilibrium}. As an additional motivation, it is worth noting, that these restricted
networks, become more meaningful if, instead of a routing game, we consider a throughput-scheduling
problem~\cite{sgall2012open}. Suppose the edges represent machines on which competing players want to maximize their
throughput. The jobs can run in parallel on multiple machines up to some maximum rate of~$d_j$ and each machine has a
maximal service rate of~$\nu_e$ and individual time horizons~$H - \tau_e$. A very similar model without the
game-theoretical aspect is for example studied by Armony and Bambos~\cite{armony2003queueing}.

In order to obtain a Nash equilibrium, it is worth noting that players do not care on which edge their flow is sent, as long
as it arrives at the destination before $H$. To model a suitable strategy, we introduce the set of \emph{active edges}
$E'(T) \coloneqq \Set{ e \in E | T_e < H}$ on which flow still arrives at $t$ before~$H$, depending on the exit times.
The resulting strategy for a player~$j \in P$ on edge~$e \in E$ could look as follows:
\begin{align}
    g_{e,j}(\theta, T) \coloneqq \begin{cases}
    \frac{\nu_{e}}{\addstackgap[1.5pt]{\scriptsize$\sum\limits_{e' \in E'(T)} \nu_{e'}$ }} &\text{ if } e \in E'(T), \\
    0 &\text{ if } E'(T) \neq \emptyset \text{ and } e \notin E'(T), \\[6pt]
    \frac{\nu_e}{\addstackgap[1.5pt]{\scriptsize$\sum\limits_{e' \in E} \nu_{e'}$}} &\text{ if } E'(T) = \emptyset.
    \end{cases} \label{eq:strategy_profile_that_leads_to_equilibrium}
\end{align}
The third case is not of importance for the player, as none of the flow sent into the network will arrive in time
anymore. As $E'(T)$ stays constant for small increases of $T$, the same is true for $g_{e,j}$. Since $\sum_{e \in E}
g_{e,j}(\theta, T) = 1$, we have~$g_j \in \StratSpace_j$.

We will show that this strategy profile leads to a Nash equilibrium. For this we first show in
\Cref{lem:payoffs_equal_opt} that the given strategy profile leads to a total payoff equal to the system optimum $\Opt$
(the value of a maximum flow over time with inflow rate $\sum_{j \in P} d_j$). Afterwards, we determine in
\Cref{lem:share_of_Opt} that the payoff for each player given the strategy profile in
\eqref{eq:strategy_profile_that_leads_to_equilibrium} is a fixed share of the total payoff. Finally, we argue that none of the players has an incentive to deviate from the given strategy profile, since
shares cannot be increased.

For the remaining of this section, let $r_e$ be the point in time when the (first) particle that arrives at $t$ at time $H$
enters edge $e$. Formally, $r_e \coloneqq \min T^{-1}_e(H)$. Hence, $T_e(\theta) < H$ for any $\theta < r_e$ and
$T_e(\theta) \geq H$ for any $\theta \geq r_e$.

\begin{lemma} \label{lem:payoffs_equal_opt}
    For the strategy profile given in \eqref{eq:strategy_profile_that_leads_to_equilibrium} the sum of the payoffs of all players equals the system optimum $\Opt$.
\end{lemma}
\begin{proof} \label{proof:payoff_equal_opt}
To prove this, we allow the network to have transit times that equal $0$.
We split the set of instances into three cases

\paragraph{Case 1:} $\sum_{j \in P} d_j \geq \sum_{e \in E} \nu_e$.\\
W.l.o.g.\ we assume $\tau_e < H$; otherwise the edge would be superfluous and could be deleted.
We have $\Opt = \sum_{e \in E} \nu_e \cdot (H -\tau_e)$ as shown in \cite{skutella2009introduction}.
Since all players route their flow proportionally to the capacities of the active edges, we have 
\[\sum_{j \in P} f_{e, j}^+(\theta) = \frac{\nu_e}{\sum_{e' \in E'(T(\theta))} \nu_{e'}} \cdot \sum_{j \in P} d_j \geq \nu_e \qquad \text{ for all } \theta \text{ with } T_e(\theta) < H.\]
With this at hand, we can state the total outflow during the whole game:
\begin{align}
    \sum\limits_{j \in P} \rho_j 
        &= \sum\limits_{e \in E} \sum_{j \in P} F_{e, j}^-(H)
    = \sum\limits_{e \in E} \int_{0}^{H} \sum_{j \in P} f_{e, j}^-(\theta) \diff \theta \nonumber \\
        &= \sum\limits_{e \in E} \Big( \int_{0}^{\tau_e} 0 \diff \theta + \int_{\tau_e}^{H} \nu_e \diff \theta \Big)
    = \sum\limits_{e \in E} (H - \tau_e) \cdot \nu_e \nonumber = \Opt. \nonumber
\end{align}

\paragraph{Case 2:} $\sum_{j \in P} d_j < \sum_{e \in E} \nu_e$ and $\tau_e = 0$ for all $e \in E$. \\
We can easily see that $\Opt = H \cdot \sum_{j \in P} d_j$ and
\[\sum_{j \in J} f_{e, j}^+(\theta) = \frac{\nu_e}{\sum_{e' \in E} \nu_{e'}} \cdot \sum_{j \in P} d_j < \nu_e \qquad \text{ for all } \theta < H.\]
Since no queues build up, all edges stay active for all $\theta \in \horizon$. Therefore, it holds that
\begin{align*}
    \sum\limits_{j \in P} \rho_j 
        = \sum\limits_{e \in E} \int_0^H \sum_{j \in P} f_{e, j}^-(\theta) \diff \theta 
    = \sum\limits_{e \in E} \sum\limits_{j \in P} H \!\cdot\! d_j \!\cdot\! \frac{\nu_e}{\sum_{e' \in E} \nu_{e'}} 
        = H \!\cdot\! \sum\limits_{j \in P} d_j
    = \Opt.
\end{align*}

\paragraph{Case 3.} $\sum_{j \in P} d_j < \sum_{e \in E} \nu_e$ and there exists an $e \in E$ with $\tau_e > 0$. \\
We assume that $\tau_{e} < H$ for all $e \in E$. Let $e^*$ be the first edge to drop out of $E'(T(\theta))$, i.e.,
$r_{e^*}$ is minimal among all $r_e$. As in the first phase no queues build up (see Case 2), we have $r_{e^*} = H -
\tau_{e^*}$, which means that $\tau_{e^*}$ is maximal among all $\tau_e$. Emphasis should be put on the fact that the
whole flow sent into the network up to time $r_{e^*}$ arrives at $t$ in time; a volume of $\sum_{j \in P}d_j
\cdot (H - \tau_{e^*})$ in total. Hence, the system optimum cannot perform any better until $r_{e^*}$.
To show that after time $r_{e^*}$, the summed payoffs correspond to the system optimum as well, we reduce the remaining
instance. First, we remove $e^*$ from the set of edges, i.e., $\hat{E} = E \setminus \set{ e^* }$. Second, we shift the time
axis $r_{e^*}$ time units back. By that we mean that the new time $0$ corresponds to $r_{e^*}$ in the former instance
and $\hat{H} = H - r_{e^*} = \tau_{e^*}$. Everything else remains untouched, in particular all queues are empty at
$r_{e^*}$. This instance is strictly smaller, which means that eventually the reduction process must end because either
we obtain $\sum_{j \in P} d_j \geq \sum_{e \in \hat{E}} \nu_e$ (Case 1), $\tau_e = 0$ for all $e \in \hat{E}$ (Case 2)
or we reach $|\hat{E}| = 1$ in which case the total payoff trivially equals $\Opt$.
\end{proof}

%The proof consists of a case distinction. If the total supply rate is larger or equal to the summed edge capacities, the
%total inflow rates equal these capacities and so do the inflow rates of a maximum flow over time. Otherwise the players'
%flow rates are distributed so evenly that no queues are building up. As soon as an edge become inactive at time $\theta
%= H -  \tau_e$ we can reduce the instance by removing this edge and shifting the time by $-\theta$. Up to this point the
%inflow into the network was equal to the inflow of a maximum flow over time and all flow arrived on time.

\begin{lemma} \label{lem:share_of_Opt}
    For the strategy profile given in \eqref{eq:strategy_profile_that_leads_to_equilibrium}, the payoff of player $j$ is given by
\[\rho_{j} = \Opt \cdot \frac{d_{j}}{\sum_{j' \in P} d_{j'}}.\]
\end{lemma}

\begin{proof}
%We recall that $T_e(\theta)< H$ for any $\theta < r_e$ and $T_e(\theta) \geq H$ for any $\theta \geq r_e$. 
For $\theta \in [0, r_e)$ we have $\sum_{j \in P} f_{e,j}^+(\theta) = \sum_{j \in P} d_j \cdot g_{e,j}(\theta, T(\theta)) > 0$.
We want to examine the outflow rates for $\phi \in [\tau_e, H)$.
With \eqref{eq:individual_outflows} and $\theta = \max T_e^{-1}(\phi) < r_e$ we obtain for all $e \in E$ and $j \in P$ that
\begin{align*}
f_{e,j}^-(\phi) &= \frac{f_{e,j}^+(\theta)}{\sum_{j' \in P} f_{e, j'}^+(\theta)} \cdot \sum_{j' \in P} f_{e, j'}^-(\phi) \\
&= \frac{d_j \cdot g_{e,j}(\theta, T(\theta))}{\sum\limits_{j' \in P}d_{j'} \cdot g_{e,j'}(\theta, T(\theta))} \cdot \sum_{j' \in P} f_{e, j'}^-(\phi)
 =\frac{d_{j}}{\sum\limits_{j' \in P} d_{j'}} \cdot \sum_{j' \in P} f_{e, j'}^-(\phi).
\end{align*}
Taking the integral over $[0,H]$, summing over all edges $e \in E$ and using \Cref{lem:payoffs_equal_opt}, yields
\[\rho_j = \sum\limits_{e \in E} F_{e, j}^-(H) 
= \frac{d_{j}}{\sum\limits_{j' \in P} d_{j'}} \cdot \sum\limits_{e \in E} \sum_{j' \in P} F_{e, j'}^-(\phi) 
= \frac{d_{j}}{\sum\limits_{j' \in P} d_{j'}}  \cdot \Opt.\]
\vspace*{-\baselineskip}\qedhere
\end{proof}

%This follows very immediately as all players use the same strategy, which get scaled by $d_j$.  Formal proofs of both lemmas
%can found in the appendix.

With these two lemmas we can finally prove the following theorem.

\begin{theorem} \label{thm:strategy_assures_known_network_share}
The strategy profile $(g_j)_{j \in P}$ given by \eqref{eq:strategy_profile_that_leads_to_equilibrium}
is a Nash equilibrium.
\end{theorem}

\begin{proof}
We want to observe what happens if one player $j^*$ deviates from her strategy in an ex ante manner. For this let $r \coloneqq \max_{e \in E} r_e$ be the point in time when the last edge becomes inactive for the modified strategy-profile and let $\rho \coloneqq \sum_{j \in P} \rho_j$ be the total amount of flow arriving in time. Clearly, $\rho \leq \Opt$.

During $[0, r]$ every non-deviating player $j \in P \setminus \set{j^*}$ sends all her flow into edges that are still active, and therefore, all this flow arrives in time. Hence, her payoff is given by $\rho_j = d_j \cdot r$.

Since no flow that enters after $r$ can possibly arrive in time, player $j^*$ can achieve at most a payoff of $d_{j^*} \cdot r$.  It follows that her share of the total amount of flow arriving in time in upper bounded by $\frac{d_{j^*}}{\sum_{j \in P} d_j}$, i.e.,
\[\rho_{j^*} \leq \rho \cdot \frac{d_{j^*}}{\sum_{j \in P} d_j} \leq \Opt \cdot \frac{d_{j^*}}{\sum_{j \in P} d_j}.\]
By \Cref{lem:share_of_Opt} the right-side equals the payoff of player $j^*$ when choosing $g_{j^*}$ as strategy. Hence, she cannot improve her payoff by deviating, which shows that $(g_j)_{j \in P}$ is indeed an equilibrium.
\end{proof}

We observe the following: As long as the players constantly choose from the set of active egdes and the summed payoff
equals $\Opt$, i.e., no queue runs dry in case of a restricted network, the players' payoffs stay the same. Hence, there
is a whole class of Nash equilibria. 

Next we want to show that this characterizes all possible Nash equilibria, which
implies that the price of anarchy is $1$.

\begin{theorem}
For atomic splittable flow over time games on parallel-edge networks where exit times are provided as information with
right-Lipschitz strategies, the price of anarchy is~$1$.
\end{theorem}
\begin{proof}
The key observation is that every player $j$'s share of the total payoff $\sum_{j' \in P} \rho_{j'}$ is at least $d_j /
\sum_{j' \in P} d_{j'}$ as long as the player only sends flow into active edges $E'(T(\theta))$. Also, the total payoff
is never decreased when a player shifts inflow from an inactive edge to an active edge, since flow sent into inactive
edges does not arrive in time. Note here that due to the flow dynamics, the cumulative outflow function of an edge
depends non-decreasingly on the cumulative inflow function of the edge (more precisely, $F_{j, e}^-(\theta) \geq \hat
F_{j, e}^-(\theta)$ for all $\theta \in [0, H)$ if $F_{j, e}^+(\theta) \geq \hat F_{j, e}^+(\theta)$ for all $\theta \in
[0, H)$). Suppose we have a strategy profile $g$ such that the total payoff $\sum_{j \in P} \rho_j$ is strictly smaller
than $\Opt$. Either one of the players' shares of the total payoff is strictly less than $d_j / \sum_{j' \in P} d_{j'}$,
so this player can improve, or all players have a share of $d_j / \sum_{j' \in P} d_{j'}$. But then there has to be an
edge~$e$ where flow is wasted, which means that the queue on $e$ is empty and the total inflow rate is strictly smaller
than the capacity for a time span of positive measure when $e$ is still active; instead, flow is sent into inactive
edges or edges with a queue, which hinders later flow to get to $t$ in time. Hence, $j$ can improve by shifting flow
from an inactive edge or edge with a queue, as this increases the total payoff but does not decrease her share. Hence,
in both cases~$g$ was not a Nash equilibrium, which shows that the price of anarchy is~$1$.
\end{proof}

\section{Further Research} \label{sec:final}

The topic opens up a multitude of further research directions. First of all, either proving or disproving the existence
of Nash equilibria for more general networks for games with information on the exit times. As the exit times do not
cover all information that might be ne\-ces\-sa\-ry for the players to react, a responding player might have an
advantage. This is especially critical in games with non-symmetric players. To illustrate this difficulty consider the
network in \Cref{fig:more_complex_network}.

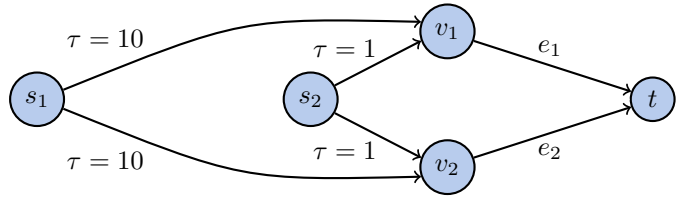
\begin{figure}[t]
    \centering
    \begin{tikzpicture}[scale=0.9,
roundnode/.style={circle, thick, draw=black, fill=mynodecolor}
]
%% LABELS
\node (transit time 11) at (0,0.9)  {$\tau = 10$};
\node (transit time 12) at (0,-0.9) {$\tau = 10$};
\node (transit time 21) at (3.5,0.75)  {$\tau = 1$};
\node (transit time 22) at (3.5,-0.75) {$\tau = 1$};
\node (e1) at (6.5,0.75)  {$e_1$};
\node (e2) at (6.5,-0.76) {$e_2$};

%% NODES 
\node[roundnode] (source1)   at (-1,0) {$s_1$};
\node[roundnode] (source2)   at (3,0) {$s_2$};
\node[roundnode] (v1)   at (5,1) {$v_1$};
\node[roundnode] (v2)   at (5,-1) {$v_2$};
\node[roundnode] (sink)     at (8,0) {$t$};

%% LINES
\draw[->, thick] (source1.20) .. controls (2,1.2)  .. (v1.160);
\draw[->, thick] (source1.-20) .. controls (2,-1.2) .. (v2.-160);
\draw[->, thick] (source2.30) -- (v1.200);
\draw[->, thick] (source2.-30) -- (v2.-200);
\draw[->, thick] (v1.-20) -- (sink.160);
\draw[->, thick] (v2.20) -- (sink.-160);

%% SPACE
\draw[draw=white] (sink) -- (7,0);

\end{tikzpicture}
    \caption{A network with two non-symmetric players. We have $t = t_1 = t_2$. Player~$1$ must commit to a split of her
    inflow rate at time $\theta$ before knowing the relevant information on the congestion on $e_1$ and $e_2$ at time
    $\theta + 10$. Player $2$ might use this to her advantage, which could lead to the non-existence of Nash equilibria.}
    \label{fig:more_complex_network}
\end{figure}

For this reason an interesting research direction would be to identify more general network
classes for which a Nash equilibrium still exists. Symmetric games where all players share the same origin and the same
destination might be a necessary restriction.

Additionally, the dependencies among different objective functions have not yet been understood very well: Shedding
light on whether the results for the maximum flow over time objective translate to a quickest flow objective (maybe via
binary-search) would be very interesting. It might even be possible to consider the average arrived flow or the average
arrival time as payoff functions, which could then be compared to an earliest arrival flow.

Finally, cooperative and
non-cooperative models might be mixed in order to assess how coalitions and selfish particles behave together (as e.g.\
in \cite{Catoni_Pallottino_1991} for the static case).

\newpage

\bibliography{literature}

%\newpage
%\appendix
%\section{Technical Proofs}
%\thmnoequilibrium*

\end{document}